\documentclass[a4paper,12pt,reqno]{amsart}

\newtheorem{theorem}{Theorem}[section]

\newtheorem{lemma}[theorem]{Lemma}
\newtheorem{proposition}[theorem]{Proposition}

\theoremstyle{definition}

\newtheorem{remark}[theorem]{Remark}

\newcommand{\vs}{\vspace{0.5 cm}}

\newcommand{\bdm}{\begin{displaymath}}
\newcommand{\edm}{\end{displaymath}}
\newcommand{\bdn}{\begin{eqnarray}}
\newcommand{\edn}{\end{eqnarray}}
\newcommand{\bay}{\begin{array}{c}}
\newcommand{\eay}{\end{array}}
\newcommand{\ben}{\begin{enumerate}}
\newcommand{\een}{\end{enumerate}}
\newcommand{\beq}{\begin{equation}}
\newcommand{\eeq}{\end{equation}}

\newcommand{\f}{\frac}
\newcommand{\F}{\mathcal{F}}

\newcommand{\form}{\F_{\alpha}}

\newcommand{\R}{\mathbb{R}}

\newcommand{\ci}{\mathbb{C}}

\newcommand{\xv}{\mathbf{x}}

\newcommand{\Xv}{\mathbf{X}}
\newcommand{\Pv}{\mathbf{P}}

\newcommand{\yv}{\mathbf{y}}

\newcommand{\rv}{\mathbf{r}}
\newcommand{\sv}{\mathbf{s}}

\newcommand{\tv}{\mathbf{t}}

\newcommand{\kv}{\mathbf{k}}

\newcommand{\pv}{\mathbf{p}}

\newcommand{\la}{\lambda}

\newcounter{remark}[section]

\newcommand{\be}{\begin{equation}}
\newcommand{\ee}{\end{equation}}
\newcommand{\ba}{\begin{eqnarray}}
\newcommand{\ea}{\end{eqnarray}}

\newcommand{\ve}{\varepsilon}

\newcommand{\bR}{{\mathbb R}}
\newcommand{\bC}{{\mathbb C}}

\newcommand{\n}{\noindent}

\input epsf

\newcommand{\donothing}[1]{}

\usepackage{bbm}

\title[]{On the quantum mechanical three-body problem with zero-range interactions}

\author[]{Giulia Basti}
\address{Dipartimento di Matematica, Sapienza Universit\`a di Roma, P.le A. Moro 5, 00185, Roma, Italy}
\email{basti@mat.uniroma1.it}

\author[]{Alessandro Teta} 
\address{Dipartimento di Matematica, Sapienza Universit\`a di Roma, P.le A. Moro 5, 00185, Roma, Italy} 
\email{teta@mat.uniroma1.it}

\begin{document}


\begin{abstract}
In this note we discuss the quantum mechanical three-body problem with pairwise zero-range interactions in dimension three. 
We review the state of the art concerning the construction of the corresponding Hamiltonian as a self-adjoint operator in the bosonic and in the fermionic case.  Exploiting a quadratic form method, we also prove  self-adjointness and  boundedness from below  in the case of three identical bosons when the Hilbert space is suitably restricted, i.e., excluding  the ``s-wave" subspace.
\end{abstract}



\maketitle

\hfill {\em Dedicated to Pavel}

\section{Introduction}\label{intr}

The  quantum mechanical three-body problem with pairwise zero-range interactions is a subject of considerable interest both for physical applications and for its peculiar mathematical structure. 

\n
The model has been introduced around the middle of the last century to describe nuclear interactions at low energy. More recently, interesting applications have been developed also   
in the physics of cold atoms, particularly in connection with the study of the Efimov effect. This is essentially due to the  experimental  possibility to realize, via the so-called Feshbach resonance, situations  where the interaction is well described by a zero-range force, in particular in the unitary limit. Roughly speaking,  unitary limit means  that the two-body interaction is characterized by a zero-energy resonance or, equivalently, by an infinite value of the scattering length.    

\n
The correct definition of the Hamiltonian, the conditions for the occurrence of the Efimov effect and the analysis of the stability problem, i.e., the existence of a finite lower bound for the Hamiltonian, have been widely studied both in the physical \cite{bh,cmp,ct,cw,cft, km,wc1,wc2} and in the mathematical \cite{CDFMT, CDFMT2, DFT,fm,FT,ms,m1,m3,m4} literature. 

\n
Here we shall review the state of the art concªthe construction of the Hamiltonian as a self-adjoint operator. Exploiting a quadratic form method, we also prove  lower boundedness  of the Hamiltonian in the case of three identical bosons when the Hilbert space is suitably restricted, i.e., excluding  the ``s-wave'' subspace.

\n
The formal Hamiltonian describing three quantum particles  in $\bR^d$, $d=1,2,3$, interacting via a zero-range, two-body interaction can be written as
\beq\label{hamnd}
\mathcal H =- \sum_{i=1}^3 \f{1}{2m_i} \Delta_{\xv_i} + \sum_{\underset{i < j}{i,j=1}}^3 \nu_{ij} \,\delta(\xv_i - \xv_j),
\eeq
where $\xv_i \in \bR^d$, $i=1, 2,3$,  is the coordinate of the $i$-th particle,   $m_i$ is the corresponding  mass, $\Delta_{\xv_i}$ is the Laplacian relative to  $\xv_i $, and $\nu_{ij} \in \bR$ is the strength  of the interaction between particles $i$ and $j$. To simplify the notation we set $\hbar =1$. 

\n
In order to give a rigorous meaning to \eqref{hamnd} as a self-adjoint operator in $L^2(\R^{3d})$, 
the first step is to give a mathematical definition, i.e., to establish the conditions that such  Hamiltonian must satisfy.  
We first notice that, in any reasonable definition,  the interaction term of the Hamiltonian must be non trivial only  on the hyperplanes $\cup_{i<j} \{\xv_i = \xv_j\}$, where the coordinates of two particles coincide. As a starting point, it is therefore natural to consider the operator $\dot{\mathcal H}_0$ defined as the free Hamiltonian restricted to a domain of smooth functions vanishing in the neighbourhood of each hyperplane $\{\xv_i=\xv_j\}$. Such operator is symmetric but not self-adjoint and one (trivial) self-adjoint extension is obviously the free Hamiltonian.  Then we define a  Hamiltonian for a system of three quantum particles in $\bR^d$ with a two-body, zero-range interaction as a non trivial self-adjoint extension of $\dot{\mathcal H}_0$.  
As a consequence of the definition, any such Hamiltonian acts as the free Hamiltonian outside the hyperplanes $\cup_{i<j} \{\xv_i = \xv_j\}$ and it is characterized by a specific  boundary condition satisfied by the wave function at each  hyperplane  $\{\xv_i = \xv_j\}$.

 \n
The second and more important step  is the  explicit construction of the self-adjoint extensions.  
The two most frequently used techniques  are Krein's theory of self-adjoint extensions and  approximation by regularized Hamiltonians, in the sense of the limit of the resolvent or of the quadratic form.  
In dimension one the problem is relatively simple due to the fact that the interaction term is a small perturbation of the free Hamiltonian in the sense of quadratic forms. In dimension two  a natural class of Hamiltonians with local zero-range interactions was constructed in \cite{DFT} and  it was also shown that such Hamiltonians are all bounded from below. In dimension three  the analysis is more delicate and in the rest of the paper we shall discuss the problem in some detail.

\n
In order to explain the difficulty, we first consider the simpler two-body  case where, in the center of mass  reference frame,  one is reduced to study a one-body problem in the relative coordinate $\xv$  with a  fixed $\delta$-interaction placed at the origin. In this case  (see, e.g., \cite{al}) the entire class of self-adjoint extensions describing Hamiltonians with point interaction 
can be explicitly constructed. One can show that the domain $D(h_{\alpha})$ of each Hamiltonian $h_{\alpha}$  consists of functions  $ \psi \in L^2(\bR^3) \cap H^2(\bR^3 \setminus \{0\})$ such that  
\beq\label{bc0}
\psi(\xv) = \f{q}{|\xv |} + r  + o(1)\, ,	\hspace{1cm}	\text{with}\;\; \;r=\alpha\, q\, ,
\eeq
for $  |\xv| \rightarrow 0$, where $q \in \ci$ and $\alpha \in \bR$ is a parameter proportional to the inverse of the scattering length. The relation $r=\alpha q$ in \eqref{bc0}  should be understood as the generalized boundary condition satisfied at the origin by all the elements of the domain. Moreover, by definition $h_{\alpha}$ satisfies
\be\label{acth}
(h_{\alpha} \psi)(\xv)= - \f{1}{2\mu} (\Delta \psi)(\xv)\,, \;\;\;\;\;\; \text{for}\;\; \xv\neq 0
\ee
where $\mu$ denotes the reduced mass of the two-body problem. 

\n
 In the three-particle case the  characterization of all possible self-adjoint extensions of $\dot{\mathcal H}_0$ is more involved. In order to circumvent the difficulty, a natural strategy is to construct  a  class of extensions  based on the analogy with the two-body case. More precisely, one considers an extension of $\dot{\mathcal H}_0$, called Skornyakov-Ter-Martirosyan (STM) operator $H_{\alpha}$,    which, roughly speaking, is a symmetric operator acting on functions $\psi \in L^2(\bR^{9}) \cap H^2( \bR^{9} \setminus \cup_{i<j} \{ \xv_i=\xv_j\})$ 
satisfying the following condition for $|\xv_i - \xv_j| \rightarrow 0$:
\beq\label{bcn}
\psi(\xv_1, \xv_2 , \xv_3)= \f{Q_{ij} (\rv_{ij}, \xv_k) }{|\xv_i - \xv_j|} + R_{ij} (\rv_{ij}, \xv_k)+ o(1)\,, \;\;\;\;\;\text{with} \;\; R_{ij}= \alpha_{ij} Q_{ij}\, ,
\eeq
where 
\be\label{rij}
\rv_{ij}= \f{m_i \xv_i + m_j \xv_j}{m_i + m_j}
\ee
$k \neq i,j$, $Q_{ij}$ is a  suitable function defined on the hyperplane $\{\xv_i = \xv_j\}$ and $\{\alpha_{ij}\} $ is  a collection of real parameters labelling the extension.  Notice that in the above limiting procedure for $|\xv_i - \xv_j| \rightarrow 0$ we keep fixed the center of mass of the particles $i,j$ and the position of the remaining particle. 
Furthermore, one has
\be \label{stm}
(H_{\alpha} \psi)(\xv_1,\xv_2,\xv_3)= (H_{f} \psi)(\xv_1,\xv_2,\xv_3)\,, \;\;\;\;\;\; \text{for}\;\; \xv_i \neq \xv_j
\ee
where $H_f$ is the free Hamiltonian. 

\n
Noticeably, the boundary condition \eqref{bcn} defining the STM extension of $\dot{\mathcal H}_0$ is a  natural  generalization to the three-body case of the condition \eqref{bc0} that characterizes  the two-body case.   Unfortunately, unlike \eqref{bc0}, \eqref{bcn} does not necessarily define a self-adjoint operator. Indeed, for a system of three identical bosons it was shown in \cite{fm} that the STM operator is not self-adjoint and all  its self-adjoint extensions are unbounded from below owing to the presence of an infinite sequence of energy levels $E_k$ going to $-\infty$ for $k \rightarrow \infty$. In \cite{MM} this result was generalized  to the case of three distinguishable  particles with different masses. This kind of instability  is known in the literature as the Thomas effect.  
It should be stressed that the Thomas effect is strongly related to the well-known Efimov effect (see, e.g., \cite{bh}) even if, to our knowledge, a rigorous mathematical investigation of this  connection is still lacking. 

\n
Here we describe an approach to the stability problem based on the theory of quadratic forms. In particular, in section 2 we explicitly construct the quadratic form naturally associated to the STM operator in the general case of three particles with different masses. 

\n
In sections 3 and 4 we consider two particular cases where the Hilbert space of states is suitably restricted, e.g., introducing  symmetry constraints on the wave function. In such cases the quadratic form is shown to be closed and bounded from below, thus defining a self-adjoint and bounded from below Hamiltonian of the system. 

\n
In the first case we consider a system of three identical bosons and we show that instability occurs only in the ``s-wave" subspace. More precisely, we  restrict the Hilbert space  to the wave functions which are not invariant under rotation of the coordinates of each particle and we prove that the quadratic form is closed and bounded from below on such subspace.

\n
In the second case we discuss the antisymmetry constraint. In fact, a wave function that is antisymmetric under exchange of coordinates of two particles necessarily vanishes at the coincidence points of such two particles, thus making their mutual zero-range interaction ineffective. Therefore, it is reasonable to expect that in a system of two identical fermions plus a different particle the interaction term in the Hamiltonian is less singular, thus making the system stable. Indeed, it has been shown that this is in fact the case for suitable values of the mass ratio (see, e.g., \cite{CDFMT, CDFMT2, m3,m4}).

\section{The energy form}

\n
We start illustrating the construction of the quadratic form in the simple case of the one-body Hamiltonian $h_{\alpha}$, formally introduced in section 1. The idea is to represent the generic element of $D(h_{\alpha})$ in the form
\be\label{dec1}
\psi=w+qg
\ee
where $w$ is a smooth function, $q \in \bC$ and
\be
g(\xv)=\f{1}{|\xv|}
\ee
The singular part $qg$ in the decomposition  \eqref{dec1} can be thought as the electrostatic potential produced by the point charge $q$ placed at the origin. According to decomposition \eqref{dec1}, the boundary condition \eqref{bc0} can be rewritten as
\be\label{bc00}
w(0)= \alpha q
\ee
Taking into account \eqref{acth} and \eqref{dec1}, the expectation value of $h_{\alpha}$ can be represented as
\begin{align}
F_{\alpha}(\psi)= (\psi, h_{\alpha} \psi) = \lim_{\ve \rightarrow 0} \int_{|\xv >\ve} \!\!\!d \xv\, \overline{\psi}(\xv) \left( -\f{1}{2\mu} \Delta \psi \right)(\xv) \nonumber \\
= \f{1}{2\mu}  \lim_{\ve \rightarrow 0} \int_{|\xv >\ve} \!\!\!d \xv\, \overline{w} (\xv) (-\Delta w)(\xv) + \f{\overline{q}}{2\mu} \lim_{\ve \rightarrow 0} \int_{|\xv >\ve} \!\!\!d \xv\, g(\xv) (-\Delta w)(\xv)
\end{align}
Integrating by parts, taking the limit $\ve \rightarrow 0$ and using \eqref{bc00}, we arrive at the following quadratic form
\be\label{for0}
F_{\alpha}(\psi)= \f{1}{2\mu} \int \! d\xv \, |\nabla w (\xv)|^2 + \f{2 \pi }{\mu} \alpha |q|^2
\ee
which is defined on the natural domain
\be\label{dofo0}
D(F_{\alpha})= \left\{ \psi \in L^2(\R^3) \,|\, \psi = w + q g, \, |\nabla w|\in L^2(\R^3), \, q \in \bC \right\}
\ee
It is a simple exercise to show that the form \eqref{for0}, \eqref{dofo0} is closed and bounded from below. Therefore it defines a self-adjoint and bounded from below operator which obviously coincides with $h_{\alpha}$. One can also notice that, defining
\be
g^{\la}(\xv)= \f{e^{- \sqrt{\la}|\xv|}}{|\xv|}\,, \;\;\;\;\;\;\la >0
\ee
 the following equivalent representation of the form domain holds
\be
D(F_{\alpha})= \left\{ \psi \in L^2(\R^3) \,|\, \psi =w^{\la} +q g^{\la}, \, w^{\la} \in H^1(\R^3) , \, q\in \bC \right\}
\ee
where $H^s(\R^d)$ denotes the standard Sobolev space in $\R^d$ of order $s \in \R$.  Accordingly one has 
\be
F_{\alpha}(\psi)= \f{1}{2\mu} \! \int \!\! d\xv \, \Big( |\nabla w^{\la}(\xv)|^2      +\la |w^{\la}(\xv)|^2 - \la |\psi(\xv)|^2     \Big) + \f{2\pi}{\mu} \left(\alpha +\sqrt{\la}\right)|q|^2
\ee

\n
In the three-particle case we follow the same idea.  We first introduce the notation $\Xv =(\xv_1,\xv_2,\xv_3)$, $\Pv=(\pv_1,\pv_2,\pv_3)$ for positions and momenta of the particles, $M=m_1+m_2+m_3$ for the total mass,  $\mu_{ij}=\f{m_i m_j}{m_i+m_j}$ for the reduced masses and $\hat{f}$ for the Fourier transform of $f$. We set $x=|\xv|$ for $\xv \in \R^3.$ Then we introduce the "potential" produced by the "charges" $Q= \{Q_{ij}\}$ distributed on the hyperplanes $\{\xv_i=\xv_j\}$. With an abuse of notation, we set
\beq
	\left(GQ\right)(\Xv) =   \sum_{i \prec j}  \left(GQ_{ij} \right)(\Xv)=     \sum_{i\prec j}\frac{1}{(2\pi)^5\mu_{ij}}\int d\Pv\,e^{i\Xv\cdot\Pv}\,\frac{\hat{Q}_{ij}(\pv_i+\pv_j,\pv_{k})}{H_f(\Pv)}
\eeq
where $k \neq i,j$, $H_f(\Pv)$ denotes the free Hamiltonian in the momentum variables and with $\prec$ we refer to the order $1\prec 2,\,2\prec 3,\,3\prec 1.$ Following the line of proposition 6.3 in \cite{FT}, one  
 shows that $GQ$ solves in the distributional sense the equation
\beq\label{hfg}
	H_f(GQ)(\Xv)=2\pi\sum_{i\prec j}\frac{1}{\mu_{ij}}Q_{ij}(\rv_{ij},\xv_k)\,\delta(\xv_i-\xv_j)
\ee
where $\rv_{ij}$ is defined in \eqref{rij}. In particular this implies
\beq\label{hg}
	H_f(GQ)(\xv_1,\xv_2,\xv_3)=0 \qquad \text{if\;\;\;\;}  \xv_i\neq\xv_j.
\eeq
Moreover $GQ$ has the following behaviour when $|\xv_i-\xv_j|\rightarrow 0$
\beq\label{bcg}
	(GQ)(\Xv)=\frac{Q_{ij}(\rv_{ij},\xv_k)}{|\xv_i-\xv_j|}-(\Gamma Q)_{ij}(\rv_{ij},\xv_k)+o(1)
\eeq
where
\begin{align}\label{gamma}	
	(\Gamma Q)_{ij}(\rv_{ij},\xv_k)\!=\!\frac{1}{(2\pi)^3}\!\!\int\! d\sv\, d\tv\, e^{i(\rv_{ij}\cdot\sv+\xv_k\cdot\tv)}\!\sqrt{\!\frac{\mu_{	 ij}}{m_i\!+\!m_j}\! s^2\!+\!\frac{\mu_{ij}}{m_k}t^2}\,\hat{Q}_{ij}(\sv,\tv)\nonumber \\
	-\frac{1}{(2\pi)^5}\int d\Pv \frac{e^{i\rv_{ij}(\pv_i+\pv_j)+i\xv_k\cdot\pv_k}}{H_f(\Pv)}\left[\frac{\hat{Q}_{ik}(\pv_i+\pv_k,\pv_j)}{\mu	 _{ik}}+\frac{\hat{Q}_{jk}(\pv_j+\pv_k,\pv_i)}{\mu_{jk}}\right]
\end{align}

\n
Proceeding in analogy with the one-body case we decompose the generic element $\psi$ in $D(H_\alpha)$ as 
\beq\label{dec}
	\Psi=u+GQ
\eeq
where $u$ is a smooth function. Then the boundary condition \eqref{bcn}, using  \eqref{bcg}, can be rewritten as 
\beq\label{bcw}
	u(\Xv)\Big|_{\xv_i = \xv_j}=(\Gamma Q)_{ij}(\rv_{ij},\xv_k)+\alpha_{ij}Q_{ij}(\rv_{ij},\xv_k)\eeq
Using the decomposition \eqref{dec}, we obtain the explicit expression of the quadratic form $\mathcal E_{\alpha}$ associated to the operator $H_\alpha.$ We set $\mathcal{D}_{\varepsilon}=\{\mathbf{X}\in\R^9 \,|\, |\mathbf{x}_i-\mathbf{x}_j|>\varepsilon, \, \;\forall \,i,j\}$. Then taking into account \eqref{stm}, \eqref{hfg} and the boundary condition \eqref{bcw} we have
\begin{align}\label{form}
		\mathcal E_{\alpha}(\Psi)=&(\Psi,H_\alpha\Psi)=\lim_{\varepsilon\to 0} \int_{\mathcal{D}_{\varepsilon}}d \Xv\,\overline{\Psi(\Xv)}(H_f\Psi)(\Xv) \nonumber\\
							 =&(u,H_f\, u)+\lim_{\varepsilon\to 0}\int_{\mathcal{D}_{\varepsilon}}d\Xv\,\overline{GQ(\Xv)}\,(H_f u)(\Xv) \nonumber \\
							 =&(u,H_f\, u)\!+\!\sum_{i\prec j}\frac{2\pi}{\mu_{ij}}\!\left[\alpha_{ij}\|Q_{ij}\|^2\!+\!\!\int\! d\rv_{ij}\,d\xv_k
								\overline{Q_{ij}(\rv_{ij},\xv_k)}(\Gamma Q)_{ij}(\rv_{ij},\xv_k)\right] \nonumber \\
							 =&(u,H_f\, u)\!+\!\sum_{i\prec j}\!\frac{2\pi}{\mu_{ij}}\!\left[\alpha_{ij}\|Q_{ij}\|^2\!+\!\!\int\! d\sv\, d\tv|\hat{Q}_{ij}(\sv
								,\tv)|^2\!\sqrt{\!\frac{\mu_{ij}}{m_i\!+\! m_j}s^2\!+\!\frac{\mu_{ij}}{m_k}t^2}\right. \nonumber \\
								&\left.-\frac{1}{(2\pi)^2\mu_{jk}}2\Re\int d\Pv\frac{\overline{\hat{Q}_{ij}(\pv_i+\pv_j,\pv_k)}\hat{Q}_{jk}(\pv_j+\pv_k,\pv
								 _i)}{H_f(\Pv)}\right]
\end{align}
where in the last equality we have used the definition of $\Gamma Q$ given in \eqref{gamma}. 
For later use, it is convenient to rewrite in a different form the last two integrals in the above formula. Let us introduce the change of variables
\begin{align}\label{coo}
	\left\{
	\begin{aligned}
		\pv&=\pv_1+\pv_2+\pv_3\\
		\kv_1&=\frac{m_j+m_k}{M}\pv_i-\frac{m_i}{M}\pv_j-\frac{m_i}{M}\pv_k\\
		\kv_2&=\frac{m_i+m_j}{M}\pv_k-\frac{m_k}{M}\pv_i-\frac{m_k}{M}\pv_j
	\end{aligned}
	\right.
\end{align}
Then defining 
\be
\hat{\zeta}_{ij}(\mathbf{k},\mathbf{p})=\hat{Q}_{ij}\left(\frac{m_i+m_j}{M}\mathbf{p}-\mathbf{k},\frac{m_k}{M}\mathbf{p}+\mathbf{k}\right)
\ee
 we have
\beq
	\int\! d\Pv\frac{\overline{\hat{Q}_{ij}(\pv_i\!+\!\pv_j,\pv_k)}\hat{Q}_{jk}(\pv_j\!+\!\pv_k,\pv_i)}{H_f(\Pv)}\!=\!\!\int\! d\pv\, d\kv_1 	d\kv_2\frac{\overline{\hat{\zeta}_{ij}(\kv_2,\pv)}\hat{\zeta}_{jk}(\kv_1,\pv)}{\frac{k_1^2}{2\mu_{ij}}\!+\!\frac{k_2^2}{2\mu_{jk}}\!+\!\frac	 {\kv_1\cdot\kv_2}{m_j}\!+\!\frac{p^2}{2M}}
\eeq
Moreover, defining the variables $\pv =\tv+\sv,$ $\kv =\frac{m_i+m_j}{M}\tv-\frac{m_k}{M}\sv$, we also have
\begin{multline}\label{diag}
	\int d\sv\, d\tv|\hat{Q}_{ij}(\sv,\tv)|^2\sqrt{\frac{\mu_{ij}}{m_i+m_j}s^2+\frac{\mu_{ij}}{m_k}t^2}=\\
	\int d\pv\, d\kv |\hat{\zeta}_{ij}(\kv,\pv)|^2\sqrt{\frac{\mu_{ij}M}{m_k(m_i+m_j)}k^{2}+\frac{\mu_{ij}}{M}p^{2}}
\end{multline}
Noticing  that $\|Q_{ij}\|=\|\hat{\zeta}_{ij}\|$, we  obtain the following equivalent expression for  $\mathcal E_{\alpha}$
\begin{align}\label{formdiv}
	\mathcal E_{\alpha}(\Psi)& = \; (u,H_f u)  \nonumber\\  
&+\sum_{i\prec j}\frac{2\pi}{\mu_{ij}}  \Bigg[ \alpha_{ij}  \| \hat{\zeta}_{ij}\|^2 
 + \sqrt{2 \mu_{ij}} \int \!\!d\pv\, d\kv \,|\hat{\zeta}_{ij}(\kv,\pv)|^2\sqrt{\frac{Mk^2}{2m_k(m_i+m_j)}  +\frac{p^2}{2M}} \nonumber \\
							&-\frac{1}{(2\pi)^2\mu_{jk}}2\Re\int \!\!d\pv\, d\kv_1\, d\kv_2\frac{\overline{\hat{\zeta}_{ij}(\kv_2,\pv)} \, \hat{\zeta}_{jk}(
							 \kv_1,\pv)}{\frac{k_1^2}{2\mu_{ij}}+\frac{k_2^2}{2\mu_{jk}}+\frac{1}{m_j}\kv_1\cdot\kv_2+\frac{p^2}{2M}}\Bigg]
\end{align}
We define the form domain as follows (see remark \eqref{re1} at the end of this section)
\beq\label{domdiv}
	D(\mathcal E_{\alpha})=\left\{\Psi\in L^2(\R^9)   \,| \,\Psi=u+\mathcal{G}_p \zeta,\, |\nabla u| \in L^2(\R^9), \zeta=\{\zeta_{ij}\},\,\zeta_{ij}\in H^{1/2}(\R^6)\right\}
\eeq
where  $\mathcal{G}_p\zeta = \sum_{i \prec j} \mathcal{G}_p \zeta_{ij} $ is given by
\be
(GQ_{ij})(\Xv)= (\mathcal{G}_p \zeta_{ij}) (\xv_i - \xv_j, \xv_k - \xv_j, \xv_{cm})\, ,\;\;\;\;\;\; \xv_{cm}=\f{m_1\xv_1 +m_2\xv_2 +m_3\xv_3}{M}
\ee
In particular
\beq
	(\widehat{\mathcal{G}_p \zeta_{ij}})(\kv_{ij},\kv_{kj},\pv)=\f{1}{\sqrt{2\pi}\mu_{ij}} \, \frac{\hat{\zeta}_{ij}(\kv_{kj},\pv)}{\frac{k_{ij}^2}{2\mu_{ij}}+\frac{k_{kj}^2}{2\mu_{jk}}+\frac{\kv_{ij}\cdot\kv_{kj}}{m_j} +\frac{p^2}{2M}}
\eeq
where with $\kv_{ij},\kv_{kj}$ we denote the conjugate variables to  $\xv_{i}-\xv_j$ and $\xv_{k}-\xv_{j}$ respectively.

\n
We remark that the dependence on the variable $\pv$ (the total momentum) in the last two integrals in \eqref{formdiv}  is essentially irrelevant. This fact can be  seen introducing a different decomposition for the elements of $D(\mathcal E_{\alpha})$. More precisely, we define $\mathcal{G}\zeta = \sum_{i \prec j} \mathcal{G} \zeta_{ij} $, where

\beq
(\widehat{\mathcal{G} \zeta_{ij}})(\kv_{ij},\kv_{kj},\pv)=\f{1}{\sqrt{2\pi}\mu_{ij}} \, \frac{\hat{\zeta}_{ij}(\kv_{kj},\pv)}{\frac{k_{ij}^2}{2\mu_{ij}}+\frac{k_{kj}^2}{2\mu_{jk}}+\frac{\kv_{ij}\cdot\kv_{kj}}{m_j} }
\eeq
and we set
\beq
\Psi= u +\mathcal G_p \zeta = v + \mathcal G \zeta\,, \;\;\;\;\;\;\;\; \Psi \in D(\mathcal E_{\alpha})
\eeq
By a direct computation we find
\begin{align}\label{formdiv2}
	\mathcal E_{\alpha}(\Psi)=&(\Psi,h_{cm}  \Psi) + ( v, h_f v) \nonumber\\
	+&\sum_{i\prec j}\frac{2\pi}{\mu_{ij}}\left[\alpha_{ij}\|\hat{\zeta}_{ij}\|^2  
	+ \sqrt{2 \mu_{ij}} \int \!\! d\pv\, d\kv \,|\hat{\zeta}_{ij}(\kv,\pv)|^2\sqrt{\frac{Mk^2}{2m_k(m_i+m_j)}  } \right. \nonumber \\
							- &  \frac{1}{(2\pi)^2\mu_{jk}}2\Re\int \!\! d\pv\, d\kv_1\, d\kv_2\frac{\overline{\hat{\zeta}_{ij}(\kv_2,\pv)} \, \hat{\zeta}_{jk}(
							 \kv_1,\pv)}{\frac{k_1^2}{2\mu_{ij}}+\frac{k_2^2}{2\mu_{jk}}+\frac{1}{m_j}\kv_1\cdot\kv_2}\Bigg]
\end{align}
where 
\beq
h_{cm} =\f{p^2}{2M}\,, \;\;\;\;\;\; h_f= H_f - h_{cm}
\eeq 
From \eqref{formdiv2} it is clear that the dependence on the variable $\pv$ is only parametric and therefore irrelevant.  
In particular, for factorized wave function $\Psi= f \cdot \psi$, where $f$ is a function of the center of mass coordinate and $\psi $ is a function of the relative coordinates, we obtain

\beq
\mathcal E_{\alpha}(\Psi) = \|\psi\|^2 (f, h_{cm} f) +  \|f\|^2 \mathcal F_{\alpha}(\psi)
\eeq
where 
\beq\label{dom3}
D(\mathcal F_{\alpha})=\Big\{ \psi \in L^2(\R^6)\,|\, \psi=w+\mathcal G \xi, \; |\nabla w| \in L^2(\R^6), \; \xi=\{\xi_{ij}\} , \; \xi_{ij}\in H^{1/2}(\R^3)  \Big\}
\eeq
\begin{align}\label{formdiv3}
	\mathcal F_{\alpha}(\psi)=& ( w, h_f w) +\sum_{i\prec j}\frac{2\pi}{\mu_{ij}}\left[\alpha_{ij}\|\hat{\xi}_{ij}\|^2  
	+ \sqrt{2 \mu_{ij}} \int \! d\kv \,|\hat{\xi}_{ij}(\kv)|^2\sqrt{\frac{Mk^2}{2m_k(m_i+m_j)}  } \right. \nonumber \\
							-&\frac{1}{(2\pi)^2\mu_{jk}}2\Re\int \! d\kv_1\, d\kv_2\frac{\overline{\hat{\xi}_{ij}(\kv_2)} \, \hat{\xi}_{jk}(
							 \kv_1)}{\frac{k_1^2}{2\mu_{ij}}+\frac{k_2^2}{2\mu_{jk}}+\frac{1}{m_j}\kv_1\cdot\kv_2}\Bigg]
\end{align}

\n
This means that, choosing the center of mass reference frame, one can reduce the analysis to the quadratic form $\mathcal F_{\alpha}$.

\n
We underline that the above construction procedure has the only aim to arrive at the definitions  \eqref{formdiv}, \eqref{domdiv} or, if one chooses the center of mass reference frame, \eqref{dom3}, \eqref{formdiv3}. Such definitions are our starting point for the rigorous construction of the Hamiltonian of the three particle system under suitable symmetry constraints. 

\begin{remark}\label{re1}
We note that in \eqref{domdiv} the choice of the charges $\zeta_{ij} \in H^{1/2}(\R^6)$ (or in \eqref{dom3} the choice $\xi_{ij} \in H^{1/2}(\R^3)$) guarantees that all terms in the square brackets of \eqref{formdiv} (or \eqref{formdiv3}) are finite. However, it is not a priori clear for which class of charges the {\em sum} of the last two  terms in the square brackets is   finite. Therefore our choice has some degree of arbitrariness and in fact, in some relevant cases, a larger class of charges must be considered (\cite{CDFMT2}).
\end{remark}

\section{Three bosons for non zero angular momentum} 
For a system of three identical bosons of unitary masses, considered  in the center of mass reference frame, the Hilbert space of states is $L^2_s(\R^6)$, i.e., the space of square-integrable functions symmetric under the exchange of particle coordinates. 
In the Fourier space, we fix a pair of coordinates $\kv_1$,$\kv_2$ defined in \eqref{coo} (with $\pv=0$), e.g., $\kv_1=\pv_1, \kv_2=\pv_3$ and then $\pv_2=-\kv_1-\kv_2$, so that the symmetry condition reads $\hat{\psi}(\kv_1,\kv_2)= \hat{\psi}(\kv_2,\kv_1)=\hat{\psi}(\kv_1,-\kv_1 - \kv_2)$. 

\n
Moreover the symmetry condition  implies that $\alpha_{ij}=\alpha$ for all $i\prec j$ and, from \eqref{bcn}, that $Q_{12}=Q_{23}=Q_{31}$ and hence $\xi_{12}=\xi_{23}=\xi_{31}=\xi.$
 Then we have the following expression for the potential
\beq
	(\widehat{\mathcal{G}\xi})(\kv_1,\kv_2)=\frac{2}{\sqrt{2\pi}}\frac{\hat{\xi}(\kv_1)+\hat{\xi}(\kv_2)+\hat{\xi}(-\kv_1-\kv_2)}{k_1^2+k_2^2+\kv_1	\cdot\kv_2}
\eeq
\n
With an abuse of notation we define the  quadratic form associated to the STM operator in the bosonic case  as
\beq
	D(\form)=\Big\{\psi\in L^2_s(\R^6)\,|\, \psi=w+\mathcal{G}\xi,\, \, |\nabla w| \in L^2_s(\R^6), \; \xi\in H^{1/2}(\R^3) \Big\}
\eeq
\beq\label{formbos}
	\form(\psi)=(w, h_f w)+ \f{12}{\pi} \, \Phi_\alpha(\xi)
\eeq
where the form $\Phi_\alpha$ acting on the charge $\xi \in D(\Phi_\alpha)=H^{1/2}(\R^3)$ is given by 
\beq	
\Phi_\alpha(\xi) = \Phi^{\textup{diag}}(\hat{\xi})+\Phi^\textup{off}(\hat{\xi})+
\alpha \int \!\! d\kv\,   | \hat{\xi}(\kv)|^2
\eeq
and the diagonal part and the off-diagonal part are defined respectively by
\begin{align}
	\Phi^\textup{diag}(f)&=  \f{ \sqrt{3} \, \pi^2 }{2} \int \!\! d\kv\, k \, |f(\kv)|^2\label{formdiag}\\
	\Phi^\textup{off}(f)&=- \int \!\! d\kv_1\, d\kv_2\,  \frac{\overline{f(\kv_1)}f(\kv_2)}{k_1^2+k_2^2+\kv_1
															\cdot\kv_2 }\label{formoff}
\end{align}
It easy to see that if one can find an $f_0$ such that $\Phi^\textup{diag}(f_0) + \Phi^\textup{off}(f_0) < 0$ then, by a scaling argument, one shows that the form  \eqref{formbos} is unbounded from below. As a matter of fact, such $f_0$ can be explicitly  constructed and it is rotationally invariant (for the proof one can follows the line of   \cite{FT}, section 4). This fact is not surprising since it is known that the STM operator is not self-adjoint and all its self-adjoint extensions are unbounded from below, showing the occurrence of the Thomas effect (\cite{fm}). 

\n
Following \cite{MM}, we define  $\mathcal H_0 = \{ \psi \in L_s^2(\R^6) \, |\, \hat{\psi} =\hat{\psi} (|\kv_1|, |\kv_2|) \}$, which is an invariant subspace for the STM operator, and we consider its orthogonal complement $\mathcal H_0^{\perp}$. In the next theorem we characterize our quadratic form in $\mathcal H_0^{\perp}$. 

\begin{theorem}\label{mainth}
	The quadratic form \eqref{formbos},\ldots,\eqref{formoff} restricted to the subspace $\mathcal H_0^{\perp}$ is bounded from below and closed for any $\alpha \in \R$.
\end{theorem}

\n
We start with some preliminaries, following the line of \cite{CDFMT}.  
Given $f \in L^2(\R^3)$, 
we consider the  expansion
\beq
	f(\kv)=\sum_{l=0}^\infty\sum_{n=-l}^l f_{ln}(k)Y_l^n(\theta,\varphi)
\eeq
where $Y_l^n$ is the the spherical harmonic of order $l, n.$ 
Using the above expansion one can obtain the following decompositions for $\Phi^\textup{off}$ and $\Phi^\textup{diag}$ (see \cite{CDFMT}, Lemma 3.1)
\begin{align}
	\begin{split}\label{phipd}
		\Phi^\textup{diag}(f)=&\sum_{l=0}^{+\infty}\sum_{n=-l}^l F^\textup{diag}(f_{ln})
	\end{split}\\
	\begin{split}\label{phipo}
	\Phi^\textup{off}(f)=&\sum_{l=0}^{+\infty}\sum_{n=-l}^lF_{l}^\textup{off}(f_{ln})
	\end{split}
\end{align} 
with $F^\textup{diag}$ and $F_{l}^\textup{off}$ acting  as
\begin{align*}
	F^\textup{diag}(g)&=  \f{ \sqrt{3} \, \pi^2}{2} \int_0^{+\infty} \!\!dk\, k^3\,|g(k)|^2\\
	F_{l}^\textup{off}(g)&=-2\pi \int_0^{+\infty}\!\! dk_1\int_0^{+\infty}\!\!dk_2\, k_1^2\, \overline{g(k_1)}k_2^2\, g(k_2)\int_{-1}^1 \!\!dy \,\frac{P_l(y)}{k_1^2+k_2^2+k_1k_2y}
\end{align*}
where $P_l$ denotes the Legendre polynomial of order $l.$ 
Proceeding as in Lemma 3.2 of \cite{CDFMT}, one  proves that
\begin{align}
	& F_{l}^\textup{off}(g) \geq 0 &\text{for l odd}\label{Fdis}\\
	&F_{l}^\textup{off}(g)\leq 0 & \text{for l even}\label{Fpari}
\end{align}
Moreover $F_{l}^\textup{off}$ can be diagonalized. Setting
\beq
	g^\sharp(k)=\frac{1}{\sqrt{2\pi}}\int dx e^{-ikx}e^{2x}g(e^x)
\eeq
we have (for details see \cite{CDFMT}, Lemma 3.3)
\begin{align}
	F^\textup{diag}(g)&=\frac{\sqrt{3} \, \pi^2 }{2}\int dk |g^\sharp(k)|^2 \label{Fdiag} \\
		F_{l}^\textup{off}(g) &=	-\int dk\, S_l(k)|g^\sharp(k)|^2\label{Foff}
\end{align}
where
\begin{equation}
	S_l(k)=
				\left\{
				\begin{aligned}
					& \pi^2 \int_{-1}^1dy\,P_l(y)\frac{\cosh\left(k\arcsin\frac{y}{2}\right)}{\cos\left(\arcsin\frac{y}{2}\right)\cosh\left(k\frac{\pi}{2}\right)}\quad \quad\quad \text{$l$ even}\\
					-& \pi^2 \int_{-1}^1dy\,P_l(y)\frac{\sinh\left(k\arcsin\frac{y}{2}\right)}{\cos\left(\arcsin\frac{y}{2} \right) \sinh\left(k\frac{\pi}{2}\right)  }\quad \quad \quad \text{$l$ odd}
				\end{aligned}
				\right.
\end{equation}
Therefore the comparison between $F_{l}^\textup{off}$ and $F^\textup{diag}$  is reduced to the study of $S_l(k).$ We first notice that $S_l(k)$ as a function of $l$ (and for any fixed $k$) is   decreasing for $l$ even and increasing for $l$ odd (see Lemma 3.5 in \cite{CDFMT}). For the estimate, we  distinguish the cases of even and odd $l$. 
\begin{lemma}\label{lemmaSl1}
	For $l$ even and any $k \in \R$
	\beq\label{slpari}
		0\, \leq S_l(k)\leq 
		\pi^2 \Bigg(\f{50}{27}\pi - \f{10}{3} \sqrt{3} + \f{\sqrt{11}}{9} - \f{10}{9} t_0 \Bigg), \;\;\;\;\;\;\;\; l \neq 0
	\eeq
	where  $t_0=\arcsin(1/\sqrt{12}) \simeq 0.293$
	and 
	\be\label{s0}
	0 \leq S_0(k) \leq 4 \pi^2.
	\ee
Furthermore, for $l$ odd and any $k\in \R$
\beq
	\pi^2\Bigg(    \frac{4}{3}\sqrt{3}-\frac{8}{\pi} \Bigg)	\leq S_l(k)\leq0 
	\eeq
\end{lemma}
\begin{proof}
	Let us consider the case $l\neq 0$ and even. The positivity of $S_l(k)$  follows  from \eqref{Fpari} and \eqref{Foff}.  Since $S_l(k)$ is decreasing in $l$, we have  $S_l(k)\leq S_2(k)$, where  $S_2(k)$  is an even function.  
An explicit integration gives 
	\beq
		\begin{split}
			S_2(0)=& \; \pi^2 \int_{-1}^1dy\,(3y^2-1)\frac{1}{2\cos\left(\arcsin \frac{y}{2}\right)}\\
						=&\; \pi^2 \int_{-\pi/6}^{\pi/6}dx\,\left(12\sin^2x-1\right)=\pi^2 \left(\frac{5}{3}\pi-3\sqrt{3}\right)
		\end{split}
	\eeq
	Let us estimate the difference  $\displaystyle{S_2(0)-S_2(k)
	}$ for any positive $k.$ We have
	\begin{align}
		S_2(0)-S_2(k)&= \pi^2 \int_{-1}^1\!\! dy\,\frac{3y^2-1}{2\cos\left(\arcsin\frac{y}{2}\right)}\left(1-\frac{\cosh\left(k\arcsin\frac{y}{2}
									\right)}{\cosh\left(k\frac{\pi}{2}\right)}\right)\\
								 &=2\pi^2 \,  I(k)
	\end{align}
	where
	\beq
		I(k)=\int_0^{\pi/6} \!\!  dt\,(12\sin^2t-1)\left(1-\frac{\cosh(kt)}{\cosh\left(k\frac{\pi}{2}\right)}\right)
	\eeq
	Since $s(t)=12\sin^2t-1$ is negative if $t<t_0$ and positive otherwise we can write
	\beq\label{I}
		\begin{split}
			I(k)&=-\int_0^{t_0}dt\,|s(t)|+\int_{t_0}^{\pi/6}dt\,s(t)\\
				&\phantom{{}={}}+\frac{1}{\cosh\left(k\frac{\pi}{2}\right)}\left[\int_0^{t_0}dt\,|s(t)|\cosh(kt)-\int_{t_0}^{\pi/6}dt\,s(t)\cosh(	
					kt)\right]\\
				&\geq \frac{1}{\cosh\left(k\frac{\pi}{2}\right)}\left[(b-a)\cosh\left(k\frac{\pi}{2}\right)+a-b\cosh\left(k\frac{\pi}{6}\right)
				 \right]
		\end{split}
	\eeq
	where $a=\int_0^{t_0}dt\,|s(t)|$ and $b= \int_{t_0}^{\pi/6}dt\,s(t)$, with $b-a>0$. Denoting 
	\beq
		g(k)=a+\left(\frac{10}{9}b-a\right)\cosh\left(k\frac{\pi}{2}\right)-b\cosh\left(k\frac{\pi}{6}\right)
	\eeq
	we can rewrite \eqref{I} as
	\beq
		I(k)\geq\frac{g(k)}{\cosh\left(k\frac{\pi}{2}\right)}-\frac{b}{9}.
	\eeq
Let us show that $g(k)\geq0.$ We have
	\beq\label{sqb}
		g'(k)=\frac{\pi}{2}\left(\frac{10}{9}b-a\right)\sinh\left(k\frac{\pi}{2}\right)\left[1-A\, \frac{3 \sinh\left(k\frac{\pi}{6}\right)}{
		\sinh\left(k\frac{\pi}{2}\right)}\right]
	\eeq
	where
	\beq
		A=\frac{b}{10b-9a}
	\eeq
	The term in square bracket in \eqref{sqb} is positive, then $g'(k)\geq 0$ which, together with $g(0)=\f{b}{9}$, implies $g(k)\geq0.$ Thus we find $S_2(0)-S_2(k)\geq - \f{2\pi^2}{9} b$. Inserting the explicit expression for $b$, we obtain the estimate \eqref{slpari}. 

\n
In the case $l=0$ the estimate \eqref{s0} is straightforward.

\n
Let us consider the case $l$ odd. 
	From \eqref{Fdis} and \eqref{Foff} it follows $S_l(k)\leq0$.  Noticing that $S_l(k)$ is an even function and it is increasing in $l$, we have $S_l(k)\geq S_1(k)$.
	
	\n
	Since $S_1(0)= \pi^2 \left(\frac{4}{3}\sqrt{3}-\frac{8}{\pi}\right) <0,$ $\lim_{k\to\infty}S_1(k)=0$ and $S'_1(k)\neq 0$  for $k>0$ we obtain the thesis.
\end{proof}
\n
The following estimate, which is the main tool in the proof of theorem \ref{mainth}, is a direct consequence of the above lemma.

\begin{proposition}\label{stima}
Let $f\in \!D(\Phi_{\alpha})$ such that $f(\kv)=\sum_{l=1}^{+\infty}\sum_{n=-l}^l f_{ln}(k)Y_l^n(\theta,\phi)$. Then 
\beq
		-\Gamma\, \Phi^\textup{diag}(f)\leq\Phi^\textup{off}(f)\leq \Lambda\, \Phi^\textup{diag}(f)
	\eeq
where 
\beq
\Gamma =\f{100}{27 \sqrt{3}} \,\pi - \f{20}{3} + \f{2 \sqrt{11}}{9\sqrt{3}} - \f{20}{9 \sqrt{3}}\, t_0 \simeq 0.101\,, \;\;\;\;\;\;\;\;\; \Lambda = -\f{8}{3 }+ \f{16}{\sqrt{3}\, \pi} \simeq 0.274
\eeq 
\end{proposition}

\begin{proof}
Using \eqref{phipo}, \eqref{Fdis}, \eqref{Foff}, \eqref{slpari}, \eqref{Fdiag}, \eqref{phipd}, we have
\begin{align}
\Phi^\textup{off}(f)&= \sum_{l=1}^{+\infty}\sum_{n=-l}^l F_l^\textup{off} (f_{ln}) \geq \sum_{\stackrel{l=2}{ l\, even}}^{+\infty}\sum_{n=-l}^l F_l^\textup{off} (f_{ln}) \nonumber\\ 
&= - \sum_{\stackrel{l=2}{ l\, even}}^{+\infty}\sum_{n=-l}^l \int\!\!dk\, S_l(k) \, |f_{ln}^\sharp (k)|^2 \geq - \Gamma   \sum_{\stackrel{l=2}{ l\, even}}^{+\infty}\sum_{n=-l}^l \f{\sqrt{3} \, \pi^2}{2}\int\!\!dk \, |f_{ln}^\sharp (k)|^2 \nonumber\\
&\geq - \Gamma \, \Phi^\textup{diag}(f)
\end{align}
and analogously one also proves the estimate $\Phi^\textup{off}(f)\leq \Lambda\, \Phi^\textup{diag}(f)$.
\end{proof}

\begin{proof}[Proof of Theorem \ref{mainth}]
	We first consider the simpler case $\alpha> 0$. From the definition \eqref{formbos} and proposition \ref{stima} we obtain the positivity of $\mathcal F_{\alpha}$ 
		\beq
		\begin{split}
			\form(\psi)&= (w, h_fw)+ \f{12}{\pi} \Phi_{\alpha}(\xi)\\
								 &\geq \f{12}{\pi}\left[\Phi^\textup{off}(\hat{\xi})+\Phi^\textup{diag}(\hat{\xi})+
									\alpha \int d\kv|\hat{\xi}(\kv)|^2\right]\\
& \geq  \f{12}{\pi}\left[(1- \Gamma) \Phi^\textup{diag}(\hat{\xi})+
									\alpha \int d\kv|\hat{\xi}(\kv)|^2\right] \geq 0
	\end{split}
	\eeq

	\n
Let us  prove the closure of $\form.$	Let $\{\psi_n\} = \{w_n+\mathcal{G}\xi_n\}$ be a sequence in $D(\form)$ such that $\psi_n \rightarrow \psi \in L^2_s(\R^6)$
	 and $\form(\psi_n-\psi_m)\to0.$ 
	
	\n
	From $\form(\psi_n-\psi_m)\to0,$ the positivity of $h_f$ and the lower bound for $\Phi_\alpha$ it follows
	\begin{align}
		&\int d\kv_1\,d\kv_2\,(k_1^2+k_2^2)\big|(\hat{w}_n-\hat{w}_m)(\kv_1,\kv_2)\big|^2\to0\\
		&\|\xi_n - \xi_m\|_{H^{1/2}} \rightarrow 0
	\end{align}
	Thus there exist $v \in L^2_s(\R^6)$ and $\xi \in H^{1/2}(\R^3)$ such that 
	\begin{align}
		&\int d\kv\, \big|\sqrt{k_1^2+k_2^2}\,w_n(\kv_1,\kv_2)-v(\kv_1,\kv_2)\big|^2\to0\\
		&\|\xi_n - \xi\|_{H^{1/2}} \rightarrow 0
	\end{align}
	Defining  $\hat{w}=\frac{\hat{v}}{\sqrt{k_1^2+k_2^2}}$, for any $\varepsilon >0$ we have  
	\begin{align}
		&\int_{\R^6_\varepsilon}d\kv_1\,d\kv_2\, |(\hat{w}_n-\hat{w})(\kv_1,\kv_2)|^2\to0\label{wn}\\
		& \int_{\R^6_\varepsilon}
		\big|\bigl(\widehat{\mathcal{G}\xi}_n-\widehat{\mathcal{G}\xi}\bigr)(\kv_1,\kv_2)\big|^2\to0\label{xin}
	\end{align}
	where $\R^d_\varepsilon=\{\xv\in\R^d\,| \, x\geq\varepsilon\}.$ 
	From \eqref{wn} and \eqref{xin} in particular we obtain
	\begin{equation*}
		\psi=w+\mathcal{G}\xi \in D(\mathcal F_{\alpha})
	\end{equation*}
	and also $\form(\psi_n-\psi)\to0.$ This concludes the proof in the case $\alpha> 0.$
	
	\n
	In order to study the case $\alpha\leq 0$ it is convenient to consider the following decomposition for the generic $\psi$ in the domain of 
	$\form$
	\beq
		\psi=w^\lambda+\mathcal{G}^\lambda\xi
	\eeq
	where $\lambda >0$ and 
	\beq
		\mathcal{G}^\lambda\xi(\kv_1,\kv_2)= \f{2}{\sqrt{2\pi}} \frac{\hat{\xi}(\kv_1)+\hat{\xi}(\kv_2)+\hat{\xi}(-\kv_1-\kv_2)}
		{k_1^2+k_2^2+\kv_1\cdot\kv_2+\lambda}
	\eeq
	Thus $\mathcal{G}^\lambda\xi$ belongs to $L_s^2(\R^6)$ and $w^\lambda$ is in $H^1(\R^6).$ Moreover  the 
	quadratic form can be rewritten as
	\beq
		\form(\psi)=(w^\lambda,h_fw^\lambda)+\lambda\|w^\lambda\|^2-\lambda\|\psi\|^2+ \f{12}{\pi} \, \Phi_\alpha^\lambda(\xi)
	\eeq
	where
	\beq
		\Phi_\alpha^\lambda(\xi)=\left[\Phi_{\lambda}^\textup{diag}(\hat{\xi})
		+\Phi_{\lambda}^\textup{off}(\hat{\xi})+\alpha\int d\kv|\hat{\xi}(\kv)|^2\right]
	\eeq
	and
	\begin{align}
		\Phi_{\lambda}^\textup{diag}(f)&= \pi^2 \int d\kv |f(k)|^2\sqrt{\frac{3}{4}k^2+\lambda}\\
		\Phi_{\lambda}^\textup{off}(f)&=-\int d\kv_1\,d\kv_2\,\frac{\overline{f(\kv_1)}f(\kv_2)}
																			 {k_1^2+k_2^2+\kv_1\cdot\kv_2+\lambda}
	\end{align}
Proceeding as in the case $\lambda =0$ (\cite{CDFMT}), one has 	
	\beq\label{normaeq}
		-\Gamma \, \Phi_{\lambda}^\textup{diag}\leq\Phi_{\lambda}^\textup{off}\leq\Lambda \, \Phi_{\lambda}^\textup{diag}
	\eeq
	Therefore the quadratic form is bounded from below
	\be
	\mathcal F_{\alpha}(\psi) \geq - \f{ \alpha^2}{\pi^4 (1 - \Gamma)^2} \, \|\psi\|^2
	\ee
	The proof that  $\mathcal F_{\alpha}$ is closed follows exactly the same line of the proof of theorem 2.1 in \cite{CDFMT} and it is omitted for the sake of brevity.
	\end{proof}
	
	\n
	We conclude observing that theorem \ref{mainth} implies the existence of a self-adjoint operator $H_{\alpha,0}^{\perp}$ in $\mathcal H_0^{\perp}$ which, at least formally, coincides with the STM operator restricted to $\mathcal H_0^{\perp}$. Such operator $H_{\alpha,0}^{\perp}$ is positive for $\alpha \geq 0$ and bounded from below by $-\f{\alpha^2}{\pi^4 (1-\Gamma)^2}$ for $\alpha <0$.

\vs
\section{System of fermions}

In a system of identical fermions the wave function, due to the antisymmetry  under exchange of coordinates, vanishes at the coincident points of any pair of particles and therefore the zero-range interaction is ineffective.  On the other hand, in physical applications it is relevant the  case of a mixture of $N $ identical fermions of one species and $M$ identical fermions of another species. Here the dynamics is non trivial since  each fermion of one species feels the zero-range interaction with all the fermions of the other species. In particular, numerical simulations seem to suggest (\cite{mp}) that the system is stable at least for mass ratio equal to one but, in this generality, no rigorous result is available (see \cite{FT} for a formulation of the problem in terms of quadratic forms). A  significant aspect of the fermionic problem is that the stability of the system depends on the value of the mass ratio. This has been explicitly shown in the case of $N\geq 2$ identical fermions of mass one plus a different particle of mass $m$. More precisely one defines   
\be
\Lambda (m,N)  =   2\pi^{-1} (N-1)    (m+1)^2 \bigg[ \f{1}{ \sqrt{m(m+2)}} - \arcsin\bigg(\f{1}{m+1}\bigg) \bigg] 
\ee

\n
For each $N$, the function $\Lambda(\cdot,N)$ is positive, decreasing and satisfies $\lim_{m\rightarrow 0} \Lambda(m,N)  
= \infty$, \\
$\lim_{m \rightarrow \infty} \Lambda(m,N)=0$. Therefore, for each $N$  the equation $\Lambda(m,N)=1$ admits exactly one solution $m^*(N)>0$, increasing with $N$ and such that $m>m^*(N)$ if and only if $\Lambda(m,N)<1$. 
Furthermore, following the strategy outlined in section 2, we consider the STM operator for this fermionic case and construct the associated quadratic form, still denoted by $\mathcal F_{\alpha}$. In \cite{CDFMT} it is proved the following result.

\begin{theorem}\label{th2}
(Stability) If $m>m^*(N)$ then   $\mathcal F_{\alpha}$ is closed and bounded from below. In particular $\mathcal F_{\alpha}$ is  positive for $\alpha \geq 0$ and  bounded from below by $-\f{\alpha^2}{4 \pi^4 (1-\Lambda(m,N))^2}$ for $\alpha <0$. 
 Therefore the corresponding STM operator $H_{\alpha}$ is self-adjoint and bounded from below, with the same lower bound. 

\n
(Instability) If  $m<m^*(2)$ then  $\mathcal F_{\alpha}$ is unbounded from below for any $\alpha \in \bR$.
\end{theorem}

\n
The above theorem provides an optimal result in the case $N=2$, i.e., stability for $m>m^*(2)$ and instability for $m<m^*(2)$, where $m^*(2) \simeq 0.0735$, in agreement with previous heuristic results in the physical literature (\cite{bh}) and also with other mathematical results (\cite{m5}, 
\cite{m4}). On the other hand, in the case $N>2$ we get only a partial result since no information is given for $m\in (m^*(2), m^*(N))$ and, in order to fill this gap, a more careful analysis of the role of the antisymmetry is required (for other results in this direction we refer to \cite{cmp}, \cite{m3}). 

\n
The special case $N=2$ in the unitary limit, i.e., for $\alpha=0$, exhibits a further interesting behavior that we want to discuss in the rest of this section.  In the center of mass reference frame we choose relative coordinates $\yv_1=\xv_1- \xv_0$, $\yv_2=\xv_2-\xv_0$, where $\xv_1, \xv_2$ are the coordinates of the fermions and $\xv_0$ denotes the coordinate of the other  particle. Let  $L_a^2(\R^6)$ be the Hilbert space of states, i.e., the space of square integrable functions anisymmetric under the exchange of coordinates. Moreover we have $\xi_{12}=0$ and the antisymmetry condition implies $\xi_{20}=-\xi_{10}:= -\xi$. Then the potential in the Fourier space takes the form
\be
(\widehat{\mathcal G \xi}) (\kv_1, \kv_2)= \f{2}{\sqrt{2\pi}} \,\f{\hat{\xi}(\kv_1) - \hat{\xi}(\kv_2)}{k_1^2 + k_2^2 + \f{2}{m+1} \kv_1 \cdot \kv_2}
\ee
and the quadratic form associated to the STM operator is
\be\label{domu}
D(\mathcal F_0)=\Big\{  \psi \in L^2_a(\R^6)\,|\, \psi= w + \mathcal G \xi ,\; |\nabla w|\in L^2_a(\R^6), \; \xi \in H^{1/2}(\R^3) \Big\}
\ee
\be\label{foru}
\mathcal F_0(\psi)= (w, h_f w) + \f{2(m+1)}{\pi m} \, \Phi_0 (\xi)
\ee
where
\begin{align}
\Phi_0(\xi) &= \Phi^{\textup{diag}}(\hat{\xi})+\Phi^\textup{off}(\hat{\xi})\\
	\Phi^\textup{diag}(f)&=  \f{ 2 \, \pi^2 \sqrt{m(m+2)} }{m+1} \int \!\! d\kv\, k \, |f(\kv)|^2\label{formdiag}\\
	\Phi^\textup{off}(f)&= \int \!\! d\kv_1\, d\kv_2\,  \frac{\overline{f(\kv_1)}f(\kv_2)}{k_1^2+k_2^2+\f{2}{m+1} \kv_1
															\cdot\kv_2 }\label{nond}
\end{align}
From theorem \ref{th2} we know that the form is closed and bounded from below for $m>m^*$ and unbounded from below for $m<m^*$ (here we have used the shorthand notation $m^*=m^*(2)$). We also notice the main differences with respect to the form in the bosonic case, i.e., the dependence on $m$ and, more important, the sign $+$ in front of the integral in  \eqref{nond}. This implies that for the estimate of $\Phi^\textup{off}(f)$ one has to study the terms for $l$ odd, and in particular the case $l=1$, in the expansion in spherical harmonics of $f$.  

\n
As a matter of fact, for suitable values of the mass $m$ the above quadratic form can be modified by enlarging the class of admissible charges and the new quadratic form turns out to be closed and bounded from below. Therefore it defines a Hamiltonian, different from $H_{\alpha}$, describing an additional three-body interaction besides  the standard two-body zero-range interaction (see \cite{CDFMT2} for details).

\n
In order to explain the above assertion, we proceed formally. Let us define
\be\label{xi-}
\hat{\xi}^-_n (\kv) = \f{1}{k^{2-s}} \, Y_{1}^n(\theta, \phi)\,, \;\;\;\;\;\;\;\; 0<s<1, \;\;\;\;\;\; n=0,\pm1
\ee
Notice that $\xi^-_n \notin L^2(\R^3)$ but this fact is not relevant since for $\alpha=0$ the condition $\Phi_0(\xi) < \infty$  does not require square-integrability of $\xi$. The crucial point is that both $\Phi^\textup{diag}(\hat{\xi}^-_n)$ and   $\Phi^\textup{off}(\hat{\xi}^-_n)$ diverge, due to the behavior of $\hat{\xi}^-_n(\kv)$ for large $k$ and the two infinities can compensate for an appropriate value of the mass. Indeed, by a direct computation one finds 
\begin{align}
\Phi_0(\xi^-_n) &= 2\pi\Bigg[ \f{\pi \sqrt{m(m+2)}}{m+1} + \int_{-1}^{1}\!\!\!\!dt\, t\!\!\int_0^{\infty}\!\!\! dq\, \f{q^s}{q^2 +1 + \f{2}{m+1}tq} \Bigg] \int_0^{\infty}\!\!\!dk\, \f{1}{k^{1-2s}} \nonumber\\
&:= g(m,s) \int_0^{\infty}\!\!\!dk\, \f{1}{k^{1-2s}} = \infty \;\;\;\;\;\;\text{unless}\;\;\;\;\;\; g(m,s)=0
\end{align}
The problem is then reduced to the study of the equation $g(m,s)=0$. One can show that for $s\in[0,1]$ there is a unique solution $m(s)$, monotonically increasing, with $m(0)=m^*$ and $m(1):= m^{**} \simeq 0.116$. For $m\in (m^*, m^{**})$ we can therefore define the inverse function $s(m)$, with $0<s(m)<1$,  which satisfies $g(m, s(m))=0$. This means that for each $m\in (m^*, m^{**})$ the charge \eqref{xi-} with $s=s(m)$ can be considered to enlarge the class of  admissible charges and to construct a more general quadratic form. Starting from the above argument, one can prove the following result.

\begin{theorem}\label{th3}
For any $m\in (m^*,m^{**})$ and $\beta := \{\beta_n\}$, $n=0,\pm1$, the quadratic form in $L^2_a(\R^6)$ 
\begin{align}
D(\mathcal F_{0,\beta}) &= \,\Big\{ \psi \in L_a^2(\R^6)\,|\, \psi = w + \mathcal G \eta, \; |\nabla w|\in L^2_a(\R^6), \; \eta \in H^{-1/2}(\R^3), \; \nonumber\\
&\;\;\;\;\;\;\;\;\eta= \xi+\sum_{n=-1}^1 q_n \xi^-_n, \; \Phi^\textup{diag}(\hat{\xi})<\infty, \; q_n \in \bC \Big\}\\
\mathcal F_{0,\beta}(\psi)&= (w, h_f w)+ \f{2 (m+1)}{\pi m} \Phi_0(\xi) + \sum_{n=-1}^1 \beta_n |q_n|^2
\end{align}
is closed and bounded from below. Then it uniquely defines a self-adjoint and bounded from below Hamiltonian $H_{0,\beta}$, $D(H_{0,\beta})$.
\end{theorem}

\n
We conclude with some comments. 

\n
i) At heuristic level, the Hamiltonian $H_{0,\beta}$ has been introduced and studied in the physical literature (see, e.g., \cite{wc1}). From the mathematical point of view, an analogous result has been found in \cite{m4} using an appoach based on the theory of self-adjoint extensions. Nevertheless, in \cite{m4} the analysis is done for $\alpha \neq 0$, which requires charges in $L^2$. Therefore the parameter $s$ in \eqref{xi-} is chosen in the interval $(0,1/2)$ and for this reason the new Hamiltonian is constructed only for a smaller range of mass, i.e., for $m\in (m^*, m^{**}_{Minlos})$, with $m^{**}_{Minlos} = m(s)|_{s=1/2} < m^{**}$.

\n
ii) The quadratic form $\mathcal F_{0, \beta}$ constructed in theorem \ref{th3} generalizes the previous one $\mathcal F_0$ in the sense that $\lim_{\beta \rightarrow \infty} \mathcal F_{0, \beta}= \mathcal F_0$. 

\n
iii) A final, and more important, comment concerns the boundary condition satisfied by an element of $D(H_{0,\beta})$. Denoting $R=\sqrt{y_1^2 + y_2^2}$ and choosing for simplicity $\xi^-_{\pm1}=0$, for $R\rightarrow 0$ one finds
\be
\psi(\yv_1, \yv_2)= \f{q_0}{R^{2+s(m)}} +\f{\nu(m) \beta_0 q_0}{R^{2-s(m)}} + o(R^{s(m)-2})
\ee
where $\nu(m)$ is a given positive function of $m$. In analogy with the case of a point interaction (see \eqref{bc0}), such boundary condition describes an interaction  supported in $\yv_1=\yv_2=0$, i.e., when the positions of all the three particles coincide. Therefore the new Hamiltonian $H_{0,\beta}$ describes the two-body (resonant) zero-range interactions plus an effective three-body point interactions. 

\vs
\section{References}


\renewcommand{\refname}{}    
\vspace*{-36pt}              

\vs

\end{document}